\documentclass{article}

\pdfoutput=1

     \usepackage[preprint]{neurips_2024}

\usepackage[utf8]{inputenc} 
\usepackage[T1]{fontenc}    
\usepackage{hyperref}       
\usepackage{url}            
\usepackage{booktabs}       
\usepackage{amsfonts}       
\usepackage{nicefrac}       
\usepackage{microtype}      
\usepackage{xcolor}         

\usepackage{graphicx} 
\usepackage{bm} 
\usepackage{lipsum} 
\usepackage{amsmath, amssymb,amsthm}
\usepackage{amssymb}
\usepackage{mathrsfs}
\usepackage{appendix}
\usepackage{epigraph}
\usepackage{algorithm,algorithmic}
\usepackage{natbib}
\setcitestyle{numbers,square}

\theoremstyle{definition}
\newtheorem{theorem}{\noindent Theorem}[section]
\newtheorem{lemma}[theorem]{\noindent Lemma}

\newtheorem{corollary}[theorem]{\noindent Corollary}
\newtheorem{definition}{\noindent Definition}[section]
\newtheorem{example}{\noindent Example}[section]
\newtheorem{remark}{\noindent Remark}[section]

\newtheorem{assumption}{\noindent Assumption}[section]

\numberwithin{equation}{section}
\numberwithin{figure}{section}
\DeclareMathOperator*{\argmax}{argmax} 

\def\advantage{\mathrm{adv}}

\title{A Payoff-Based Policy Gradient Method in Stochastic Games with Long-Run Average Payoffs}
\author{%
  Junyue Zhang  \\
  Academy of Mathematics and Systems Science\\
  Chinese Academy of Sciences\\
  Beijing, BJ 100190 \\
  \texttt{zhangjunyue@amss.ac.cn} \\
  \And
  Yifen Mu  \\
  Academy of Mathematics and Systems Science\\
  Chinese Academy of Sciences\\
  Beijing, BJ 100190 \\
  \texttt{mu@amss.ac.cn} \\
}

\begin{document}

\maketitle

\begin{abstract}
    Despite the significant potential for various applications, stochastic games with long-run average payoffs have received limited scholarly attention, particularly concerning the development of learning algorithms for them due to the challenges of mathematical analysis. In this paper, we study the stochastic games with long-run average payoffs and present an equivalent formulation for individual payoff gradients by defining advantage functions which will be proved to be bounded. This discovery allows us to demonstrate that the individual payoff gradient function is Lipschitz continuous with respect to the policy profile and that the value function of the games exhibits the gradient dominance property. Leveraging these insights, we devise a payoff-based gradient estimation approach and integrate it with the Regularized Robbins-Monro method from stochastic approximation theory to construct a bandit learning algorithm suited for stochastic games with long-run average payoffs. Additionally, we prove that if all players adopt our algorithm, the policy profile employed will asymptotically converge to a Nash equilibrium with probability one, provided that all Nash equilibria are globally neutrally stable and a globally variationally stable Nash equilibrium exists. This condition represents a wide class of games, including monotone games.
\end{abstract}

\section{Introduction}
\par Ever since they were proposed by Shapley \cite{Shapley} in the 1950's, stochastic games have been extensively studied with a large amount of applications in fields such as multi-agent reinforcement learning \cite{Multiagentlearning}, robotics \cite{robotics}, autonomous driving \cite{driving}. Unlike static games, in stochastic game settings, the games will be played in many rounds. Before choosing their actions, players can know the current state which determines the rules of the game for that stage in advance. After the players act in this round, they will receive their own instantaneous rewards. At the same time, the state of the game will move to the next state based on the transition probabilities induced by the current state and the players' actions. Therefore, in contrast to matrix games, in stochastic games, each player has to balance two objectives: maximizing the one-step gains or optimizing the long-run payoffs.

\par In general, there are four classic long-run payoff models in stochastic games \cite{Stochasticgamessolan}: total payoffs in finite horizon, total payoffs in random stopping time frameworks, discounted payoffs in infinite horizon, and average payoffs in infinite horizon. Regardless of the specific long-run payoff function utilized, these functions are significantly influenced by the transition probabilities that are a direct consequence of the players' strategic choices. As a result, they often exhibit a considerable degree of nonlinearity, even when the action space available to players in each state is finite. Therefore, some well-known algorithms such as fictitious play \cite{FP} and no-regret learning \cite{ASimpleAdaptiveProcedureLeadingtoCorrelatedEquilibrium}, which are widely used in matrix games, cannot be simply applied to learning in stochastic games.

\par However, as demonstrated in \cite{deng_complexity_2022}, calculating a Nash equilibrium for stochastic games with discounted payoffs is classified as PPAD-Complete, a complexity equivalent to that of matrix games \cite{PPADhardfinite}. This parallel suggests that it may be possible to adapt learning algorithms for finding Nash equilibria in some special stochastic games, or correlated equilibria in all such games, analogous to methods employed in matrix games. The literature on this topic is extensive; however, the majority of studies have concentrated on finite horizon payoffs \cite{Vlearning, WhenCanWeLearnGeneralSumMarkovGames} or discounted payoffs in infinite horizons \cite{Nashq, GradientPlayinStochasticGamesStationaryPointsandLocalGeometry}, with less emphasis on random stopping time payoffs \cite{OntheconvergenceofpolicygradientmethodstoNashequilibriaingeneralstochasticgames} and average payoffs over infinite horizons \cite{LearningStationaryNashEquilibriumPoliciesin}. Nevertheless, as highlighted in \cite{ReinforcementLearningAnIntroduction}, the first two models possess inherent limitations. In certain scenarios, the model of average payoffs in infinite horizon may more accurately reflect real-world conditions, particularly when players are required to take actions at a high frequency within a short time frame.

\par In our study, we will focus on stochastic games with long-run average payoffs, and the main contributions of our work are as follows:
\begin{enumerate}
    \item We extend the concept of advantage functions from reinforcement learning\cite{Policygradientmethodsforreinforcementlearningwithfunctionapproximation} to stochastic games with long-run average payoffs, and prove that it is bounded and well-defined, thereby laying a solid foundation for further analysis.
    \item We prove that the individual payoff gradients in stochastic games with long-run average payoffs are Lipschitz continuous, as researchers have done for the other three types of stochastic games \cite{daskalakis2021independent}\cite{PolicygradientmethodsfindtheNashequilibriuminNplayergeneralsumlinearquadraticgames}\cite{Onlineconvexoptimizationinthebanditsettinggradientdescentwithoutagradient} \cite{GradientPlayinStochasticGamesStationaryPointsandLocalGeometry}. Furthermore, the value functions possess gradient dominance property, so in stochastic games with aforementioned payoffs, all first-order stationary policies are Nash equilibrium, vice verse.
    \item Capitalizing on our observations, we develop a payoff-based gradient estimation approach inspired by simultaneous perturbation stochastic approximation method \cite{BanditlearninginconcaveN-persongames} and integrate it with the Robbins-Monro method\cite{RobbinsMonro} and the mirror descent algorithm \cite{Beck2003MirrorDA} to construct a bandit learning algorithm suited for stochastic games with long-run average payoffs. Our algorithm is distributed, relatively simple, and can be applied to lots of games. What the players only need is instantaneous rewards they receive in games.
    \item We prove that using our algorithm, the learning process will converge to a Nash equilibrium with probability $1$ if all Nash equilibria are globally neutrally stable and a globally variationally stable Nash equilibrium exists.
\end{enumerate}

\par The paper is organized as follows. In Section $2$, we provide an overview of the fundamental concepts associated with stochastic games. In Section $3$, we analyse the properties of the value function and demonstrate that the individual payoff gradients are Lipschtiz continuous. In Section $4$, we introduce a payoff-based methods for estimating the individual payoff gradients, and present a mirror descent algorithm for learning the Nash equilibrium. Our algorithm is derived based on these methodologies. In Section $5$, we introduce the concept of stability in stochastic games including neutrally stable and variationally stable, and show the convergence of the progress induced by our algorithm. Discussions are given in Section $6$.

\section{Problem Setup and Preliminaries}
\par Throughout this paper, we focus on stochastic games involving a finite set of players, denoted by $N = \{1, 2, \ldots, n\}$, and a finite set of states $S$. Each player $i \in N$ possesses a finite set of actions $A_i$ available in every state $s \in S$. The Cartesian product $A = \prod_{i \in N} A_i$ represents the set of all possible joint actions, while $A_{-i} = \prod_{j \neq i} A_j$ signifies the set of all joint actions except player $i$. Upon reaching a state $s$, player $i$ chooses an action $a_{i} \in A_i$ and receives an immediate reward $r_i(a_i, a_{-i})$. Subsequently, the game transitions to a new state $s' \in S$ with probability $\mathbb{P}[s' | s, a]$, where $a$ represents the joint action $(a_1, \dots, a_n)$. Formally, we can use a tuple $\mathcal{G} = (S, N, (A_i)_{i \in N}, \mathbb{P}, (r_i)_{i \in N})$ to denote such a stochastic game.

\par The game will be played repeatedly as follows. At any discrete time $t = 0,1,2,\dots$, all players observe their current state $s^t$ and choose her action $a^t_i$ from $A_i$. Then they receive their immediate reward $r_i(s^t, a^t)$, denoting the joint action as $a^t = (a^t_i)_{i \in N}$. After that, the state of the game will move to $s^{t+1}$ according to the transition probability $\mathbb{P}(s^{t+1}| s^t, a^t)$, and players choose their actions in the next state $s^{t+1}$.

\par At each discrete time step $t$, for any player $i$, we define the history $\mathcal{H}^t_i$ as the collection of all information available to player $i$, encompassing their realized states, actions, and rewards. This history is formally represented as the set
\[
    \mathcal{H}^t_i = \{ s^l, a^l_i, r_i(s^l, a^l) : l = 0, 1, \ldots, t - 1 \} \cup \{ s^t \}
\]
Players will determine their actions at time $t$ based on $\mathcal{H}^t_i$. A general policy for player $i$ is characterized by a mapping $\pi_i: \mathcal{H} \to \Delta(A_i)$, where $\mathcal{H}$ is the set of all histories and $\Delta(A_i)$ represents the set of all probability distributions over the set $A_i$. Upon observing the history $\mathcal{H}^t_i$, player $i$ will execute the mixed strategy $\pi_i(\mathcal{H}^t_i)$. However, employing a policy contingent on the history can be intricate. In practice, there is a preference for more straightforward policies, particularly stationary policies, which are the focus of this paper and are defined subsequently.

\begin{definition}[Stationary policy]
    A policy $\pi_i$ for player $i$ is stationary if it is solely dependent on the current state, i.e., for any history $\mathcal{H}^t_i$, we have $\pi_i(\mathcal{H}^t_i) = \pi_i(s^t)$. And we will use $\pi_i(a_i| s)$ to denote the probability of player $i$ taking action $a_i$ in state $s$. Concurrently, we use $\Pi_i$ to denote the set of all stationary policies of player $i$. Furthermore, the set of all stationary policy profiles is denoted by $\Pi = \prod_{i \in N} \Pi_i$.
\end{definition}

\par Here we note that the stationary policy set $\Pi_i$ can be represented as a subset of $
\mathbb{R}^{|S| \times |A_i|}$. More precisely, we have
\begin{align*}
    \Pi_i = \{ (\pi_i(a_i|s))_{(a_i,s) \in A_i \times S}: \sum_{a_i \in A_i} \pi_i(a_i|s) = 1 ,\forall s \in S, \pi_i(a_i|s) \geqslant 0, \forall (a_i,s) \in A_i \times S \}.
\end{align*}

\par Given a initial state $s^0$ and a stationary policy profile $\pi$, a Markov chain on the set of states can be naturally induced. The $(s,s')$-element of the transition matrix $P_{\pi}$ is
\begin{equation}\label{transition}
    P_{\pi}(s' | s) = \sum_{a \in A} \mathbb{P}[s' | s,a] (\prod_{i=1}^{n} \pi_i(a_i|s)).
\end{equation}
Thus, assuming that the Markov chain $P_{\pi}$ is ergodic, if we let $p_{\pi}$ be the unique stationary distribution of $P_{\pi}$, we have $\lim_{t \to \infty} \mathbb{P}(s^t = s |s^0 = s_0, \pi) = p_{\pi}(s)$. And for the convenience of our analysis, we will use the following standard assumption used in the MDP literature \cite{OnlineMarkovDecisionProcesses}.
\begin{assumption}\label{assumptionintransition}
    For any stationary policy profile $\pi$ chosen by the players, the induced Markov chain $P_{\pi}$ is ergodic, and its mixing time is uniformly bounded above by a constant $\tau > 0$, that is
    \begin{align*}
        \| (w-w') P_{\pi} \|_1 \leqslant e^{-\frac{1}{\tau}} \| w-w' \|_1,\quad  \forall i \in N, \pi \in \Pi, w,w' \in \Delta(S).
    \end{align*}
\end{assumption}
This is an assumption that holds in a wide array of scenarios. For instance, if there exists a constant $\varepsilon > 0$ such that for all state $s$, $s'$ and for all joint action $a$, the transition probability satisfies $\mathbb{P}[s' | s, a] > \varepsilon$. Under these conditions, for any policy profile $\pi$, the induced transition matrix $P_{\pi}$ ensures that $P_{\pi}(s'|s) > \varepsilon$. It can be readily demonstrated that Assumption \ref{assumptionintransition} is valid in this context.
\par The objective for each player $i$ in stochastic games with long-run average payoffs is to choose a stationary policy $\pi_i$ to maximize her expected long-run average payoffs, which is given by
\begin{align}
    V_i(\pi_i, \pi_{-i}) = \mathbb{E}\bigg[\lim_{T\to \infty} \frac{1}{T} \sum^{T}_{t=0} r_i(s^t,a^t)\bigg].
\end{align}
With Assumption \ref{assumptionintransition}, we have
\begin{align*}
    V_i(\pi_i,\pi_{-i}) = \sum_{s \in S} p_{\pi_i, \pi_{-i}}(s) \sum_{a} (\prod_{i=1}^{n} \pi_i(a_i|s)) r_i(s,a_i,a_{-i}),
\end{align*}
which shows that $V_i(\pi_i, \pi_{-i})$ is well-defined, i.e., convergent if all players take stationary policies.
\par Now we can define Nash equilibrium in stochastic games with long-run average payoff as follows.
\begin{definition}[Nash equilibrium]
    A policy profile $\pi^* = (\pi^*_i)_{i \in N}$ is a Nash equilibrium for the stochastic game $\mathcal{G}$ with long-run average payoffs if for any player $i$ and any stationary policy $\pi_i$,
    \begin{align*}
        V_i(\pi^*_i, \pi^*_{-i}) \geqslant V_i(\pi_i, \pi^*_{-i}).
    \end{align*}
\end{definition}

\section{Properties of Individual Payoff Functions}
\label{PropertiesofIndividualPayoffFunctions}

\par In this section, we analyze the properties of the long-run average payoff function \( V_i(\pi_i, \pi_{-i}) \) and its associated individual payoff gradient \( \nabla_i V_i(\pi_i, \pi_{-i}) \). For notational simplicity, the individual payoff gradient may occasionally be represented by \( v_i(\pi_i, \pi_{-i}) \), and the vector of all players' individual payoff gradients by \( v(\pi) = (v_i(\pi_i, \pi_{-i}))_{i \in N} \). 
\par Inspired by \cite{Policygradientmethodsforreinforcementlearningwithfunctionapproximation}, the advantage functions of a state-action pair $(s,a)$ for player $i$ given a policy $\pi$ are defined as:
\begin{equation}
    \advantage^{\pi}_{i}(s,a) = \sum_{t=0}^{\infty} \mathbb{E}\big[r_i(s^t,a^t) - V_i(\pi_i, \pi_{-i}) | s_0 = s, a_0 = a , \pi \big],
\end{equation}
which represents the sum of the differences between the rewards received by player $i$ and the value function $V_i(\pi_i, \pi_{-i})$, given that the game begins at state $s$ with action $a$ and subsequently follows the stationary policy profile $\pi$.
\par We also define the average advantage functions as:
\begin{align}
    \overline{\advantage}^{\pi}_{i}(s,a_i) := \mathbb{E}_{a_{-i} \sim \pi_{-i}(\cdot | s)} \big[ \advantage^{\pi}_i(s, (a_i,a_{-i}))  \big].
\end{align}
\par Before embarking on the subsequent analysis, we need to clarify that $\advantage^{\pi}_{i}(s,a)$ and $\overline{\advantage}^{\pi}_{i}(s,a_i)$ is well-defined, which means they will not be infinity.
\begin{lemma}\label{advbound}
    The advantage functions $\advantage^{\pi}_i(s,a)$ are bounded with respect to the policy profile $\pi$, and so are the average advantage functions $\overline{\advantage}^{\pi}_{i}(s,a_i)$.
\end{lemma}

\par We now provide an equivalent formulation of the individual payoff gradients $v_i(\pi)$ which will help us to analyse the boundedness and smoothness of the gradients. We start with the following versions of the policy gradient theorem for stochastic games with long-run average payoffs.
\begin{theorem}[Policy gradient theorem]\label{expression}
    For any player $i$ and any stationary policy profile $\pi$, we have
    \begin{equation}
        \nabla_i V_i(\pi) = \sum_{s \in S} p_{\pi}(s) \sum_{a_i \in A_i} (\nabla_i \pi_i(a_i|s)) \overline{\advantage}^\pi_i(s,a_i).
    \end{equation}
\end{theorem}

\begin{remark}\label{relationaboutnabla}
    As we have $v_i(\pi) = \nabla_i V_i(\pi) = (\frac{\partial}{\partial \pi_i(a_i|s)} V_i(\pi))_{(a_i,s) \in A_i \times S}$, it can be obtained that
    \begin{equation}
        \frac{\partial}{\partial \pi_i(a_i|s)} V_i(\pi) = p_{\pi}(s) \overline{\advantage}^\pi_i(s,a_i).
    \end{equation}
\end{remark}

\par By Theorem \ref{expression} and Lemma \ref{advbound}, it can be immediately demonstrated that $v_i$ is bounded with respect to policy profile $\pi$.

\begin{lemma}
    For any player $i$, the individual payoff gradient $v_i(\pi)$ is bounded with respect to the policy profile $\pi$.
\end{lemma}

\par However, it is insufficient to apply the method that will be presented in the subsequent section to estimate $v_i(\pi)$. It is necessary to demonstrate that $v_i(\pi)$ is Lipschitz continuous with respect to $\pi$. This will be done by proving the following theorem based on some auxiliary lemmas, which will be provided in the appendix.
\begin{theorem}\label{Lipschitzofgradient}
    For any player $i$ and state-action pair $(s,a_i)$, the partial derivative of the value function $\frac{\partial}{\partial \pi_i(a_i|s)} V_i(\pi)$ is Lipschitz continuous with respect to $\pi$. Consequently, the individual payoff gradient $\nabla_{i} V_i(\pi)$ is also Lipschitz continuous.
\end{theorem}

\par At the end of this section we prove that, stochastic games with long-run average payoff satisfy the gradient dominance property. Before presenting the property, we need to define the mismatch coefficient as $C_{\mathcal{G}} = \max_{\pi,\pi' \in \Pi} \{ \| \frac{p_{\pi}}{p_{\pi'}} \|_{\infty} \} $. Like the prior work \cite{OntheconvergenceofpolicygradientmethodstoNashequilibriaingeneralstochasticgames}, it is be assumed that $C_{\mathcal{G}}> 0 $ and is finite in our work. Then we have

\begin{theorem}[Gradient dominance property]\label{Gradientdominanceproperty}
    For any policy profile $\pi = (\pi_i)_{i \in N} \in \Pi$, we have
    \begin{equation}
        \begin{aligned}
            V_i(\pi'_{i}, \pi_{-i}) - V_i(\pi_{i}, \pi_{-i}) \leqslant C_{\mathcal{G}} \langle \nabla_i V_{i}(\pi), \pi'_i - \pi_i \rangle,
        \end{aligned}
    \end{equation}
    where $\pi'_i \in \Pi_i$ is any proper individual deviation of player $i$.
\end{theorem}

\par Based on the gradient dominance property, in stochastic games with long-run average payoffs, Nash equilibria are equivalent to the first-order stationary points, which is similar to the case in games with discounted payoffs \cite{GradientPlayinStochasticGamesStationaryPointsandLocalGeometry}.
\begin{theorem}[First-order stationary policies are Nash]\label{FirstorderstationarypoliciesareNash}
    A policy profile $\pi^* = (\pi^*_i)_{i \in N}$ is Nash equilibrium if and only if it satisfies the first-order stationary condition
    \begin{equation}
        \begin{aligned}
            \langle v(\pi^*), \pi -\pi^* \rangle \leqslant 0,\quad  \forall \pi \in \Pi.
        \end{aligned}
    \end{equation}
\end{theorem}

\section{The Learning Framework}\label{TheLearningFramework}
\par In this section, we present our algorithm in stochastic games. It is based on a fundamental learning framework known as the regularized Robbins-Monro template \cite{mertikopoulos_unified_2024}. However, applying this method to stochastic games necessitates the estimation of the individual payoff gradients, i.e., \( v(\pi) = (v_i(\pi))_{i \in N} \), which is challenging to compute directly. To overcome this difficulty, we employ the simultaneous perturbation stochastic approximation method \cite{BanditlearninginconcaveN-persongames},  which requires only the individual payoffs for estimating \( v_i(\pi) \). By integrating these two approaches, we devise an algorithm aimed at attaining the Nash equilibrium for stochastic games with long-run average payoffs. Notably, in our algorithm, players only need to observe their own actions and resulting payoffs to learn their individual payoff gradients, so it is a distributed bandit algorithm.
\subsection{Regularized Robbins-Monro process}
\par Combining Robbins-Monro algorithm \cite{RobbinsMonro}, a famous stochastic approximation method, and ``follow-the-regularized-leader'' family of algorithms of Shalev-Shwartz \& Singer \cite{OnlineLearningandOnlineConvexOptimization}, we obtain the following learning framework called regularized Robbins-Monro template \cite{mertikopoulos_unified_2024}:
\begin{equation}
\begin{aligned}
    Y^{t+1} &= Y^t + \gamma^t \hat{v}^{t}, \\
    \pi^{t+1} &= Q(Y^{t+1})
\end{aligned}
\end{equation}
where
\begin{enumerate}
    \item $\pi^t = (\pi^t_i)_{i \in N}$ is the players' policy profile at time $t$.
    \item $\hat{v}^t = (\hat{v}^t_{i})_{i \in N}$ is an individual ``gradient-like'' signal that estimates the individual gradient $v(\pi^t) = (v_i(\pi^t))_{i \in N}$.
    \item $Y^t = (Y^t_{i})_{i \in N}$ is the weighted sum of $ \{ \hat{v}^s | s \leqslant t \}$ in dual space of $\Pi$.
    \item $ \gamma^t > 0$ is the step-size, and to guarantee the perpetual update of parameters throughout the iterative process, we assume that $\sum_{t} \gamma^t = \infty$.
\end{enumerate}
The most important element is the function $Q: \mathcal{Y} \to \Pi \subseteq \mathbb{R}^{|S| \times |A|}$, where $\mathcal{Y}$ is the dual space of $\Pi$. It is a ``general projection'' map that mirrors gradient steps in $\mathcal{Y}$ to policy space $\Pi$, and we call it the players' mirror map, which is related to the mirror descent algorithm, as decribed in reference \cite{darzentas_problem_1984}.
\par To define the mirror map more specifically, we decompose it into the product form $Q = (Q_i)_{i \in N}$, where $Q_i: \mathcal{Y}_i \to \Pi_i \subseteq \mathbb{R}^{|S| \times |A_i|}$ is the mirror map of a single player $i$. We begin with introducing the concept of ``regularizer'' on $\Pi_i$ as follows:
\begin{definition}[Regularizer]
    For any player $i$, $h_i: \mathbb{R}^{|S|\times |A_i|} \to \mathbb{R}\cup \{ \infty \}$ is a regularizer on $\Pi_i$ if
    \begin{enumerate}
        \item $h_i$ is supported on $\Pi_i$, i.e., $\{ x_i \in \mathbb{R}^{|S|\times |A_i|}: h_i(x_i) < \infty \} = \Pi_i$.
        \item $h_i$ is continuous and $K_i$-strongly convex on $\Pi_i$, i.e., there exists a constant $K_i > 0 $ such that for all $\pi_i$, $\pi'_i \in \Pi_i$ and all $ \lambda \in [0,1] $
        \begin{align*}
            h_i \left( \lambda \pi_i + \left(1-\lambda \right) \pi'_i \right) \leqslant \lambda h_i\left(\pi_i\right) + \left(1-\lambda\right) h_i\left(\pi'_i\right) - \frac{1}{2} K_i \lambda \left(1-\lambda\right)\| \pi'_i - \pi_i \|^2_2.
        \end{align*}
    \end{enumerate}
\end{definition}
And we define the mirror map $Q: \mathcal{Y}_i \to \Pi_i$ and convex conjugate $h_i^*: \mathcal{Y}_i \to \mathbb{R}$ induced by regularizer $h_i$:
\begin{align}
    Q_i(Y_i) =& \argmax_{\pi_i \in \Pi_i} \{ \langle Y_i, \pi_i \rangle - h_i(\pi_i) \}, \label{definitionofQ} \\
    h_i^*(Y_i) =& \max_{\pi_i \in \Pi_i} \{  \langle Y_i, \pi_i \rangle - h_i(\pi_i)  \}.\label{definitionofh}
\end{align}
For convenience, we denote $h(\pi) = \sum_{i \in N} h_i(\pi_i)$ for the players' aggregate regularizer and $Q = (Q_i)_{i \in N}$ for the induced mirror map.
\begin{example}[Entropic regularizaion]
    Let $h_i(\pi_i) = \sum_{a_i \in A_i, s \in S} \pi_i(a_i|s) \log \pi_i(a_i|s)$ be the (negative) Gibbs-Shannon entropy on $\Pi_i$. With straightforward calculations, the induced mirror map of each player $i$ is the logit choice map, i.e.,
    \begin{equation}
        \begin{aligned}
            (Q_i(Y_i))_{a_i,s} = \frac{\exp(Y_i(a_i,s))}{\sum_{a_i \in A_i} \exp(Y_i(a_i,s))},
        \end{aligned}
    \end{equation}
    which is Hedge algorithm in learning in finite games \cite{Cesa-Bianchi_Lugosi_2006}.
\end{example}

\par For the analysis of the convergence of the learning process, we introduce the Fenchel coupling:
\begin{equation}\label{definitionofF}
    \begin{aligned}
        F(p, y) = h(p) + h^*(y) - \langle y,p \rangle,
    \end{aligned}
\end{equation}
which can be seen as the global energy function if the equilibrium point satisfies variational stability \cite{mertikopoulos_unified_2024}. In more specific terms, it can be employed to measure the distance between a Nash equilibrium $\pi^*$ and any policy profile $\pi$, as it has the following properties.
\begin{lemma}(\cite{mertikopoulos_learning_2019})\label{propertyofFenchel}
    If $h$ is a $K$- strongly convex regularizer on $\mathcal{X}$, let $Q$ be the mirror map induced by $h$, fix some $p \in \mathcal{X}$, and for all $y,y' \in \mathcal{Y}$, the dual space of $\mathcal{X}$, we have
    \begin{enumerate}
        \item $F(p,y) \geqslant \frac{1}{2} K \| Q(y) - p \|^2.$
        \item $F(p,y') \leqslant F(p,y) + \langle y'-y, Q(y) -p \rangle + \frac{1}{2K} \| y'-y \|^2 $.
    \end{enumerate}
\end{lemma}
And we assume that
\begin{equation}\label{reciprocitycondition}
    F(p, y_n) \to 0 \text{ whenever } Q(y_n ) \to p,
\end{equation}
which is called ``reciprocity condition'' in the theory of Bregman function \cite{reciprocitycondition}.

\subsection{Simultaneous perturbation stochastic approximation}

\par Now, what we need is to estimate the gradients. However, directly estimating the individual payoff gradients $v_i(\pi^t)$ from the historical information presents a formidable challenge for players. Nevertheless, under the Assumption \ref{assumptionintransition}, players are able to efficiently evaluate their policies and approximate the payoff function $V_i(\pi_i, \pi_{-i})$ provided they adhere to a consistent strategy $\pi_i$ over an extended sequence of moves. 
\par Once we get the approximate value of $V_i(\pi)$, following \cite{Spall} \cite{Onlineconvexoptimizationinthebanditsettinggradientdescentwithoutagradient} \cite{BanditlearninginconcaveN-persongames}, players can use the simultaneous perturbation stochastic approximation approach that allows them to estimate their individual payoff gradients based on their own payoffs only. In detail, this estimation process can be summarized in the following steps for each player $i \in N$:
\begin{enumerate}
    \item Choose a query radius $\delta > 0$.
    \item Determine the policy $\pi_i \in \Pi_i$ where player $i$ hopes to estimate her payoff gradient.
    \item Draw a vector $z_i$ from the unit sphere $\mathbb{S}_i$ of $ \mathbb{R}^{|S| \times |A_i|}$ uniformly and play $\hat{\pi}_i = \pi_i + \delta z_i$.
    \item Get the approximate value $\hat{V}_i$ of value function $V_i(\pi_i, \pi_{-i})$ and obtain
    \begin{equation}\label{SPSAestimator}
        \begin{aligned}
            \hat{v}_i = \frac{d_i}{\delta} \hat{V}_i z_i, \quad d_i \text{ is the dimension of } \Pi_i.
        \end{aligned}
    \end{equation}
\end{enumerate}
Using Stoke's theorem, we can show that $\hat{v}_i$ is an unbiased estimator of the individual gradient of the $\delta$-averaged payoff function:
\begin{align}\label{deltaaverage}
	V^{\delta}_i (\pi_i, \pi_{-i}) = \frac{1}{vol(\delta \mathbb{B}_i)\prod_{j\neq i}vol(\delta\mathbb{S}_j)} \int_{\delta \mathbb{B}_i} \int_{\prod_{j\neq i} \delta\mathbb{S}_j} u_i(\pi_i + w_i, \pi_{-i} + z_{-i}) \mathrm{d} z_1 \cdots \mathrm{d} w_i \cdots \mathrm{d} z_N,
\end{align}
where $\mathbb{B}_i$ is the unit ball of $ \mathbb{R}^{|S| \times |A_i|}$.
And the Lipschitz continuity of $v_i$ we have proved can guarantee that $\| \nabla_i V_i - \nabla_{i} V^{\delta}_i \|_{\infty} = \mathcal{O}(\delta) $.
\begin{lemma}\label{estamitorproperty}
    The estimator $\hat{v} = (\hat{v}_i)_{i \in N}$ given by \eqref{SPSAestimator} satisfies
    \begin{equation*}
        \begin{aligned}
            \mathbb{E}[\hat{v}_i] = \nabla_i V^{\delta}_i (\pi_i, \pi_{-i}),
        \end{aligned}
    \end{equation*}
    with $V^{\delta}_i (\pi_i, \pi_{-i})$ as in \eqref{deltaaverage}. Moreover, we have $\| \nabla_i V_i - \nabla_{i} V^{\delta}_i \|_{\infty} = \mathcal{O}(\delta) $ if $ \nabla_i V_i $ is Lipschitz continuous.
\end{lemma}

\par Before we use the method above, there are some issues to clarify. The first thing is that as the perturbation direction $z_i$ is drawn from the unit sphere $\mathbb{S}_i$, it may fail to be tangent to $\Pi_i$, especially by our definition, all individual policies $\pi_i$ are in the surface of $\Pi_i$. Thus we need to build a equivalent representation of the policy $\pi_i$ which lie in a low-dimensional space and satisfy our need. Let
\begin{equation}
    \begin{aligned}
        \mathcal{X}_i = \{ x_i \in \mathbb{R}^{|S| \times (|A_i|-1)}: 0 \leqslant \sum_{a_i \in A_{i,last}} x_i(a_i,s) \leqslant 1, \forall s \in S, x_i \geqslant 0 \}.
    \end{aligned}
\end{equation}
where $A_{i,last}$ means the set $A_i$ deletes its last action. There is a natural bijection $M$ between $\Pi_i$ and $\mathcal{X}_i$, i.e., $M(\pi_i)(a_i, s) = x_i(a_i,s), \forall a_i \in A_{i,last}, \forall s \in S $. And it is easy to check that
\begin{equation}
    \begin{aligned}
        \frac{\partial}{\partial \pi_i(a_i|s)} V_i(\pi_i, \pi_{-i}) = & \frac{\partial}{\partial x_i(a_i,s)} V_i(x_i, x_{-i}), \forall a_i \in A_{i,last}, \forall s \in S, \\
        \frac{\partial}{\partial \pi_i(a_{i,last}|s)} V_i(\pi_i, \pi_{-i}) = & - \sum_{a_i \in A_{i,last}} \frac{\partial}{\partial x_i(a_i,s)} V_i(x_i, x_{-i}),
    \end{aligned}
\end{equation}
where $a_{i,last}$ is the action deleted in the set $A_{i,last}$.
So we can use $\mathcal{X}_i$ to represent the set of player $i$'s stationary policies and the interior of $\mathcal{X}_i$ is a closed convex set. If we let $x_i$ being in the interior of $\mathcal{X}_i$ and draw $z_i$ from $\mathbb{S}_i \subseteq \mathbb{R}^{|S| \times (|A_i|-1)}$, we can guarantee that $z_i$ is tangent to $\mathcal{X}_i$ and estimate the individual payoff gradient $\nabla_i V_i(x_i, x_{-i})$ using $\hat{x}_i = x_i + \delta z_i$. Subsequently, we can derive the gradient $\nabla_i V_i(\pi_i, \pi_{-i})$ from $\nabla_i V_i(x_i, x_{-i})$.
\par On the other hand, even when $z_i$ is a feasible direction of perturbation, the query point $\hat{x}_i = x_i + \delta z_i$ may not be in $\mathcal{X}_i$. To address this issue, we adopt the concept of a ``safety net``, as introduced in the work of \cite{agarwal2010optimal}.
\par Let $\mathbb{B}_{r_i}(p_i)$ be an $r_i$-ball centered at $p_i \in \mathcal{X}_i$ so that $\mathbb{B}_{r_i}(p_i) \subseteq \mathcal{X}_i$. Then instead of using the direction $z_i$, we consider the feasibility adjustment
\begin{equation}
    \begin{aligned}
        w_i = z_i - r_i^{-1} (x_i - p_i),
    \end{aligned}
\end{equation}
and player $i$ plays $\hat{x}_i = x_i + \delta w_i$. Noting that
\begin{equation*}
    \begin{aligned}
        \hat{x}_i = x_i + \delta w_i = (1 - r_i^{-1} \delta)x_i + r_i^{-1} \delta(p_i + r_i z_i),
    \end{aligned}
\end{equation*}
when $\delta/r_i < 1$, we have $\hat{x}_i \in \mathcal{X}_i$.
\subsection{A distributed learning algorithm in stochastic games}
\par After the preparation above, we introduce the following learning algorithm in stochastic games.
\begin{algorithm}[H]
	\caption{Multi-agent Distributed Learning Algorithm for Player $i$}     
	\label{1} 
	\begin{algorithmic}\label{algorithm}     
        \REQUIRE step-size $\gamma^t$, query radius $\delta^t$, safety ball $\mathbb{B}_{r_i}(p_i)$, time threshold $T^t$, mirror map $Q$
        \STATE Choose initial policy $\pi^0_i \in \Pi_i$ ($ x^0_i \in \mathcal{X}_i$ equivalently).
        \STATE set $Y^0_i = 0$.
        \FOR{ period $t = 0,1,2,\dots$ }
            \STATE draw $z^t_i$ uniformly from $\mathbb{S}^{|S|(|A_i|-1)}$
            \STATE let $w^t_i = z^t_i - r_i^{-1}(x_i^t - p_i)$
            \STATE play policy $\hat{x}^t_i = x^t_i + \delta^t w^t$ in the following $T^t$ stage, and get the one-step reward $\hat{V}^t_i$ at time $T^t+1$
            \STATE get the individual gradient $\hat{v}_i^t$ by the estimator of $\nabla_i V_i(x_i, x_{-i})$, in details
            \begin{equation}\label{definitionofhatv}
                \hat{v}_i^t (a_i,s) = 
                \begin{cases}
                    \frac{|S|(|A_i|-1)}{\delta^t} \hat{V}_i^t z^t_i(a_i,s) & \text{ if } a_i \neq a_{i,last}, \\
                    - \sum_{a'_i \neq a_i} \frac{|S|(|A_i|-1)}{\delta^t} \hat{V}_i^t z^t_i(a_i,s) & \text{ if } a_i = a_{i,last}.
                \end{cases}
            \end{equation}
            \STATE update $Y^{t+1}_i = Y^{t}_i + \gamma^t \hat{v}_i^t$
            \STATE update $\pi^{t+1}_i = Q(Y^{t+1}_i)$, and corresponding $x^{t+1}_i$
        \ENDFOR
	\end{algorithmic}
\end{algorithm}
In practice, each player's step size can vary, but for convenience, we assume that all players have the same step size in this paper.
\par The reason why players need a time threshold $T^t$ in the algorithm is that from any initial state $s$, after the players play a fixed policy $\hat{x}^t$ in few rounds $T^t$, the instantaneous reward $\hat{V}^t_i$ for player $i$ at time $T^t + 1$ will be close to the value function $V_i(\hat{x}^t)$ if the probability distribution over $S$ is not far from the stationary distribution $p_{\hat{x}^t}$ at that time.
\begin{lemma}\label{differenceinvalue}
    If Assumption \ref{assumptionintransition} holds and players play games by the Algorithm \ref{algorithm}, we have
    \begin{equation}
        \begin{aligned}
            | \mathbb{E}[\hat{V}^t_i] - V_i(\hat{x}^t) | \leqslant |S| (\max_{s,a} r_i(s,a)) e^{-\frac{T^t}{\tau}}.
        \end{aligned}
    \end{equation}
\end{lemma}

\section{Convergence Analysis and Results}
\label{Convergenceanalysisandresults}

\par If all players use the Algorithm \ref{algorithm} in stochastic games, we can construct the following stochastic process to describe their behaviors:
\begin{equation} \label{coreequation}
\begin{aligned}
    \hat{x}^t &= M(\pi^t) + \delta^t w^t,\\
    Y^{t+1} &= Y^t + \gamma^t \hat{v}^t, \\
    \pi^{t+1} &= Q(Y^{t+1}).
\end{aligned}
\end{equation}
In the above, $w^t = (w^t_i)_{i \in N}$ and $\hat{v}^t = (\hat{v}^t_i)_{i \in N}$ have the following expressions respectively:
\begin{align}
    w^t_i = z^t_i - r_i^{-1}(x^t_i - p_i), \quad \hat{v}^t_i = \frac{|S|(|A_i|-1)}{\delta^t} \hat{V}^t_i \cdot F_i z^t_i,
\end{align}
where
\begin{align*}
    F_i =& \left(
    \begin{array}{ccc}
         G_i &  & \\
         & \ddots & \\
         &  & G_i
    \end{array}
    \right)_{(|S| \times |A_i|) \times (|S| \times (|A_i|-1))}, \\
    G_i =& \left(
    \begin{array}{cccc}
        1 & & \\
        &\ddots& \\
         &  & 1\\
         -1 & \ldots & -1
    \end{array}
    \right)_{|A_i| \times (|A_i|-1)}.
\end{align*}
\par For the estimator $\hat{v}^t_i$, as we have proved in Lemma \ref{estamitorproperty}
\begin{equation*}
    \mathbb{E}_{z_i^t} [\frac{|S|(|A_i|-1)}{\delta^t} \hat{V}^t_i \cdot z^t_i] = \nabla_i V^{\delta}_i (x_i, x_{-i}),
\end{equation*}
we obtain that
\begin{equation}
    \mathbb{E}[\hat{v}^t_i] = \mathbb{E}_{z_i^t} [\frac{|S|(|A_i|-1)}{\delta^t} \hat{V}^t_i \cdot F_i z^t_i] = F_i \nabla_i V^{\delta}_i (x_i, x_{-i}).
\end{equation}
Thus, we can write
\begin{equation}\label{decomposition}
    \begin{aligned}
        \hat{v}^t_i =& \nabla_{i} V_i(\pi^t_i, \pi^t_{-i}) + (F_i \nabla_{i} V_i^{\delta^t}(x^t_{i,\delta^t}, x^t_{-i,\delta^t})- \nabla_{i} V_i(\pi^t_i, \pi^t_{-i})) \\
        &+ F_i (\frac{|S|(|A_i|-1)}{\delta^t} V_i(\hat{x}^t) \cdot z^t_i -\mathbb{E}_{z^t_i}[ \frac{|S|(|A_i|-1)}{\delta^t} V_i(\hat{x}^t) \cdot z^t_i])
        \\
        &+ F_i ( \frac{|S|(|A_i|-1)}{\delta^t} \hat{V}^t_i \cdot z^t_i - \frac{|S|(|A_i|-1)}{\delta^t} V_i(\hat{x}^t) \cdot z^t_i) \\
        \triangleq & \nabla_{i} V_i(\pi^t_i, \pi^t_{-i})  + b^t_i + U^t_i + \epsilon^t_i,
    \end{aligned}
\end{equation}
where in the stochastic approximation process, $b^t_i$ and $\epsilon^t_i$ are bias terms, while $U^t_i$ is the noise term.
\par Having completed the preceding preparations, we are now in a position to analyze the convergence of the process \eqref{coreequation}. As this is a learning process in games, it is hoped that it will converge to some Nash equilibrium. However, as it is of PPAD-complete complexity to find a Nash equilibrium even in finite games \cite{PPADhardfinite} and stochastic games with discounted payoffs \cite{deng_complexity_2022}, it is therefore not to be expected that our algorithm will converge to some Nash equilibrium in all stochastic games with long-run average payoffs. However, the system will converge to the Nash equilibrium globally if we assume that the game have some great properties.

\begin{definition}[Neutrally stable\cite{mertikopoulos_unified_2024}]
    A policy profile $\pi^*$ is globally neutrally stable if
    \begin{equation}
         \langle v(\pi), \pi - \pi^*  \rangle \leqslant 0, \forall \pi \in \Pi,
    \end{equation}
    where $v(\pi) = (v_i(\pi))_{i \in N}$ is the individual payoff gradient. Furthermore, if the equality holds only when $\pi$ is a Nash equilibrium, we say the policy profile $\pi^*$ is globally variationally stable.
\end{definition}

\par A classic type of games is monotone games \cite{Rosen} which can be defined as follows in stochastic games.
\begin{example}[Monotone games]\label{monotone}
    A stochastic game $\mathcal{G}$ is a monotone game if 
    \begin{equation}
        \langle v(\pi) - v(\pi') , \pi - \pi' \rangle \leqslant 0, \forall \pi,\pi' \in \Pi_i.
    \end{equation}
    Based on the first-order stationary property of Nash equilibria, any Nash equilibrium in monotone games is globally neutrally stable.
\end{example}
\par The following convergence theorem holds when Nash equilibria possess the aforementioned stability conditions.
\begin{theorem}\label{convergencetheorem}
    Suppose that in stochastic games, all Nash equilibria are globally neutrally stable and a globally variationally stable Nash equilibrium exists. If all players employ Algorithm \ref{algorithm} with parameters having the following property
    \begin{align*}
        \lim_{t \to \infty} \gamma^t =& \lim_{t \to \infty} \delta^t = 0, \quad \sum_{t=0}^{\infty} \gamma^t = \infty, \quad \sum_{t=0}^{\infty} \gamma^t \delta^t < \infty, \quad \sum_{t=0}^{\infty} (\frac{\gamma^t}{\delta^t})^2 < \infty, \quad
        \sum_{t=0}^{\infty} \frac{\gamma^t}{\delta^t} e^{- \frac{T^t}{\tau}} < \infty,
    \end{align*}
    then the sequence of realized actions $\hat{\pi}^t$ converges to Nash equilibrium with probability $1$.
\end{theorem}

\par It is shown \cite{Rosen} that, strict monotone games admit a unique Nash equilibrium and it surely possesses the globally variationally stability. So we have the following corollary for strict monotone games.
\begin{corollary}
    The sequence of realized actions $\hat{\pi}^t$ converges to the Nash equilibrium with probability $1$ in strict monotone games if the assumption of Theorem \ref{convergencetheorem} holds.
\end{corollary}

\par The proof of Theorem \ref{convergencetheorem} is divided into three steps as follows. First, for any globally neutrally stable Nash equilibrium $\pi^*$, we will demonstrate that $F(\pi^*, Y^t)$ converges to a finite random variable $F^{\infty}$. Then we will show that there exists a subsequence $\{ \pi^{t_k} \}$ which will converges to some Nash equilibrium if there exists a globally stable Nash equilibrium. Then the proof of Theorem \ref{convergencetheorem} can be completed with these two facts as well as the properties of the Fenchel coupling. The details of the proof are provided in the appendix.

\par It is necessary to choose the appropriate parameters when implementing the algorithm. We can consider parameters of the following form $\gamma^t = \gamma/ t^p$, $\delta^t = \delta/t^q$ with $\gamma$, $\delta > 0$ and $0 < p,q \leqslant 1$. In order to satisfy the assumptions of the Theorem \ref{convergencetheorem}, we need $p+q > 1$ and $p-q > 1/2$. For $T^t$, when $\tau$ is known to the players, we can let $T^t = [T\log t]+1$ with $T>0$ and $p - q + T/\tau > 1$. Otherwise, we can let $T^t = [t^T]+ 1$ with $T > 0$. A suitable pair of parameters is $(\gamma^t, \delta^t, T^t) = (1, 1/3, \sqrt{t})$.

\section{Discussion}
\par In this work, we demonstrate that the individual payoff gradients for stochastic games with long-run average payoffs exhibit Lipschitz continuity and fulfill the gradient dominance property, a feature that is also observed in the context of discounted payoffs. Based on these facts observed, we have developed a promising learning algorithm predicated on policy gradient estimation. The algorithm ensures that, provided the games possess a globally variationally stable Nash equilibrium with all Nash equilibria globally neutrally stable, the system will converge to some Nash equilibrium with probability one, if all players employ this approach.
\par However, our current analysis only confirms the global asymptotic convergence of the algorithm for certain special cases within game theory. The investigation into non-asymptotic convergence properties is still pending. Additionally, the applicability of the algorithm to a broader spectrum of stochastic games, including zero-sum games, is an open question. Moreover, it is yet to be determined if the algorithm exhibits a high probability of converging to the Nash equilibrium when the initial policy profile is in close proximity to it.

\bibliographystyle{plain}
\bibliography{ref}

\appendix

\section{Proof of Section \ref{PropertiesofIndividualPayoffFunctions}}

\begin{proof}[Proof of Lemma \ref{advbound}]
    We denote that $e_s \in \Delta(S)$ with its $s$-th element being $1$ and others being $0$, and $R^{\pi}_i$ as the expected one-step reward vector, i.e., $R^{\pi}_{i}(s) = \sum_{a} (\prod_{i=1}^{n} \pi_i(a_i|s)) r_i(s,a_i,a_{-i})$. We can observe
    \begin{equation}\label{advantage}
        \begin{aligned}
        \advantage^{\pi}_i(s,a) =& \sum_{t=0}^{\infty} \mathbb{E}[r_i(s^t,a^t) - V_i(\pi) | s_0 = s, a_0 = a, \pi] \\
        =& r_i(s,a) - V_i(\pi) + \sum_{t=1}^{\infty} [e_s P^t_{\pi} R^{\pi}_i - p_{\pi} R^{\pi}_i] \\
        =& r_i(s,a) - p_{\pi} R^{\pi}_i + \sum_{t=1}^{\infty} (e_s P^t_{\pi} - p_{\pi} P^t_{\pi}) R^{\pi}_i
        \end{aligned}
    \end{equation}
    Thus, we have
    \begin{align*}
        | \advantage^{\pi}_{i}(s,a) | \leqslant & \sum_{t=0}^{\infty} \| (e_s- p_{\pi}) P_{\pi}^t \|_{1} \| R^{\pi}_i \|_{\infty} \\
        \leqslant & \|e_s- p_{\pi}\|_{1} (\max_{s,a} r_i(s,a) )\sum_{t=0}^{\infty} e^{-\frac{t}{\tau}} \\
        \leqslant & 2 (\max_{s,a} r_i(s,a) )\sum_{t=0}^{\infty} e^{-\frac{t}{\tau}} < \infty
    \end{align*}
    As for the average advantage functions, we get
    \begin{align*}
        | \overline{\advantage}^{\pi}_{i}(s,a_i)  | = &  |\mathbb{E}_{a_{-i} \sim \pi_{-i}(\cdot | s)} \big[ \advantage^{\pi}_i(s, (a_i,a_{-i}))  \big]| \\
        \leqslant & \max_{s,a} \advantage^{\pi}_i(s,a) \\
        \leqslant & 2 (\max_{s,a} r_i(s,a) )\sum_{t=0}^{\infty} e^{-\frac{t}{\tau}} < \infty.
    \end{align*}
\end{proof}

\begin{proof}[Proof of Theorem \ref{expression}]
    By observation, we have
    \begin{equation}
        \begin{aligned}
            \nabla_i \bigg[ 
            \sum_{a_i \in A_i} \pi_i(a_i| s) \overline{\advantage}^\pi_i(s, a_i)
    \bigg]
    =\sum_{a_i \in A_i} \big( \nabla_i \pi_i(a_i|s)   \overline{\advantage}^\pi_i(s, a_i)\big) + \sum_{a_i \in A_i} 
 \pi_i(a_i|s) \big(  \nabla_i   \overline{\advantage}^\pi_i(s, a_i)\big).
        \end{aligned}
    \end{equation}
    And we get
    \begin{equation}
        \begin{aligned}
            \nabla_i   \overline{\advantage}^\pi_i(s, a_i) =& \nabla_i \bigg[
            \sum_{a_{-i} \in A_{-i}} \pi_{-i}(a_{-i}| s) \advantage^{\pi}_i(s,(a_i,a_{-i}))
            \bigg], \\
            =& \sum_{a_{-i} \in A_{-i}} \pi_{-i}(a_{-i}| s) \nabla_i  \advantage^{\pi}_i(s,(a_i,a_{-i})).
        \end{aligned}
    \end{equation}
    At the same time, we can observe
    \begin{equation}\label{advdecomposition}
    \begin{aligned}
        \advantage^{\pi}_i(s,(a_i,a_{-i})) =& \sum_{t=0}^{\infty} \mathbb{E}\bigg[ 
        r_i(s^t,a^t) - V_i(\pi) | s^0 = s, a^0 = (a_i,a_{-i}), \pi
        \bigg], \\
        =& r_i(s,a_i,a_{-i}) - V_i(\pi) + \sum_{s' \in S} \mathbb{P}(s' | s, (a_i,a_{-i}))\sum_{a' \in A} \pi(a'|s') \advantage^{\pi}_{i}(s',a'), \\
        =& r_i(s,a_i,a_{-i}) - V_i(\pi) + \sum_{s' \in S} \mathbb{P}(s' | s, (a_i,a_{-i}))\sum_{a'_i \in A_i} \pi_i(a'_i|s') \overline{\advantage}^{\pi}_{i}(s',a'_i),
    \end{aligned}
    \end{equation}
    thus we have
    \begin{equation}
        \begin{aligned}
            \nabla_i  \advantage^{\pi}_i(s,(a_i,a_{-i})) =& -\nabla_i V_i(\pi) + \sum_{s' \in S} \mathbb{P}(s' | s, (a_i,a_{-i})) \nabla_i \bigg[ \sum_{a'_i \in A_i} \pi_i(a'_i|s') \overline{\advantage}^{\pi}_{i}(s',a'_i) \bigg].
        \end{aligned}
    \end{equation}
    Combining the equations above, we obtain that
    \begin{equation}
        \begin{aligned}
            \nabla_i V_i(\pi) =& \sum_{a_i \in A_i} \big( \nabla_i \pi_i(a_i|s)   \overline{\advantage}^\pi_i(s, a_i)\big) \\
            &+ \sum_{a \in A} \pi(a | s) \sum_{s' \in S} \mathbb{P}(s' | s, (a_i,a_{-i})) \nabla_i \bigg[ \sum_{a'_i \in A_i} \pi_i(a'_i|s') \overline{\advantage}^{\pi}_{i}(s',a'_i) \bigg] \\
            &- \nabla_i \bigg[ 
            \sum_{a_i \in A_i} \pi_i(a_i| s) \overline{\advantage}^\pi_i(s, a_i)
    \bigg],
        \end{aligned}
    \end{equation}
    Summing both sides over the stationary distribution $p_{\pi}$, and by the property of stationary distributions, we have
    \begin{equation*}
        \begin{aligned}
            \sum_{s\in S} p_{\pi}(s) \nabla_i V_i(\pi) =& \sum_{s\in S} p_{\pi}(s) \sum_{a_i \in A_i} \big( \nabla_i \pi_i(a_i|s)   \overline{\advantage}^\pi_i(s, a_i)\big) \\
            &+ \sum_{s\in S} p_{\pi}(s) \sum_{a \in A} \pi(a | s) \sum_{s' \in S} \mathbb{P}(s' | s, (a_i,a_{-i})) \nabla_i \bigg[ \sum_{a'_i \in A_i} \pi_i(a'_i|s') \overline{\advantage}^{\pi}_{i}(s',a'_i) \bigg] \\
            &- \sum_{s\in S} p_{\pi}(s) \nabla_i \bigg[ 
            \sum_{a_i \in A_i} \pi_i(a_i| s) \overline{\advantage}^\pi_i(s, a_i)
    \bigg]\\
    =& \sum_{s\in S} p_{\pi}(s) \sum_{a_i \in A_i} \big( \nabla_i \pi_i(a_i|s)   \overline{\advantage}^\pi_i(s, a_i)\big) \\
            &+ \sum_{s'\in S} p_{\pi}(s')  \nabla_i \bigg[ \sum_{a'_i \in A_i} \pi_i(a'_i|s') \overline{\advantage}^{\pi}_{i}(s',a'_i) \bigg] \\
            &- \sum_{s\in S} p_{\pi}(s) \nabla_i \bigg[ 
            \sum_{a_i \in A_i} \pi_i(a_i| s) \overline{\advantage}^\pi_i(s, a_i)
    \bigg] \\
    =& \sum_{s\in S} p_{\pi}(s) \sum_{a_i \in A_i} \big( \nabla_i \pi_i(a_i|s)   \overline{\advantage}^\pi_i(s, a_i)\big).
        \end{aligned}
    \end{equation*}
    It can be demonstrated that the lemma is true without delay, given that $\nabla_i V_i(\pi)$ is independent of state $s$.

\end{proof}

\par To substantiate Theorem \ref{Lipschitzofgradient}, it is essential to cite the following lemma:

\begin{lemma}\label{prodx}
    For any real vectors $x$, $y \in \mathbb{R}^n$, with $x_i, y_i \in [0,1], \forall i = 1, \dots,n$, we have \[| \prod_{i=1}^{n} x_i - \prod_{i=1}^{n} y_i | \leqslant (2^n -1 ) \| x-y \|_{\infty}.\]
\end{lemma}

\begin{proof}
    By direct calculation, we have
    \begin{align*}
        | \prod_{i=1}^{n} x_i - \prod_{i=1}^{n} y_i | =& | \prod_{i=1}^{n} (x_i - y_i + y_i) - \prod_{i=1}^{n} y_i | \\
        =& | \prod_{i=1}^{n} (x_i - y_i) + \sum_{j=1} y_j \prod_{i \neq j} (x_i - y_i) + \cdots + \sum_{j=1}(x_j - y_j) \prod_{i \neq j} y_j  | \\ 
        \leqslant & | \prod_{i=1}^{n} (x_i - y_i)| + \sum_{j=1} y_j \prod_{i \neq j} (x_i - y_i)| + \cdots + |\sum_{j=1}(x_j - y_j) \prod_{i \neq j} y_j  | \\ 
        \leqslant& (2^n -1) \max_{i \in N}|x_i - y_i| \\
        \leqslant& (2^n -1) \| x-y \|_{\infty}.
    \end{align*}
\end{proof}

\par In order to demonstrate the Lipschitz continuity of $v_i(\pi)$, it is first necessary to establish the following property of the transition matrices $P_{\pi}$.

\begin{lemma}\label{Pt-Pt'}
    For all $t \geqslant 1$ and $w$, $w' \in \Delta(S)$, if $\pi$, $\pi'$ are any policy profiles and induce the transition matrices $P_{\pi}$, $P_{\pi'}$, we have
    \begin{equation}
        \begin{aligned}
            \| (w - w' ) ( P^{t}_{\pi} - P^{t}_{\pi'} ) \|_{1} 
            \leqslant & (2^N -1) t e^{- \frac{t-1}{\tau}} |A| |S| \| \pi - \pi' \|_{\infty} \| w - w' \|_{1}\\
            \triangleq & C t e^{- \frac{t-1}{\tau}}\| \pi - \pi' \|_{\infty} \| w - w' \|_{1}
        \end{aligned}
    \end{equation}
\end{lemma}

\begin{proof}
    First, we prove that the one-step transition matrix $P_{\pi}(s' | s)$ is Lipschitz continuous about policy profiles $\pi$, by definition \eqref{transition} and Lemma \ref{prodx}, we have
    \begin{equation}\label{transitionmatrixLip}
    \begin{aligned}
        | P_{\pi}(s' | s) - P_{\pi'}(s' | s)  | \leqslant & \sum_{a \in A} \mathbb{P}[s' | s, a] | \prod_{i=1}^n \pi_i(a_i| s) -  \prod_{i=1}^n \pi'_i(a_i| s) \rvert \\
        \leqslant & \sum_{a\in A} \mathbb{P}[s' | s, a] \cdot (2^{N}-1) \max_{i \in N} | \pi_i(a_i | s) - \pi'_i(a_i|s) | \\
        \leqslant & (2^N -1) |A| \| \pi - \pi' \|_{\infty}.
    \end{aligned}
    \end{equation}
    So, we can get that
    \begin{equation}
        \begin{aligned}
            \| P_{\pi} - P_{\pi'}  \|_{1} \leqslant (2^N -1) |A| |S| \| \pi - \pi' \|_{\infty}
        \end{aligned}
    \end{equation}
    Then we prove the lemma by induction, for $t = 1$,
    \begin{equation}
        \begin{aligned}
            \| (w - w' ) ( P_{\pi} - P_{\pi'} ) \|_{1} \leqslant & \| w - w'  \|_1 \|  P_{\pi} - P_{\pi'}  \|_{1} \\
            \leqslant & (2^N -1) |A| |S| \| \pi - \pi' \|_{\infty} \| w - w'  \|_1.
        \end{aligned}
    \end{equation}
    For $t \geqslant 2$, if the lemma is true for $t-1$, then we have
    \begin{equation}
        \begin{aligned}
            \| (w - w' ) ( P^{t}_{\pi} - P^{t}_{\pi'} ) \|_{1} \leqslant & \| (w - w' ) ( P^{t}_{\pi} - P^{t-1}_{\pi}P_{\pi'} ) \|_{1} + \| (w - w' ) ( P^{t-1}_{\pi}P_{\pi'} - P^{t}_{\pi'} ) \|_{1} \\
            = & \| (w - w' ) P^{t-1}_{\pi} ( P_{\pi} - P_{\pi'} ) \|_{1} + \| (w - w' ) ( P^{t-1}_{\pi} - P^{t-1}_{\pi'} ) P_{\pi'} \|_{1} \\
            \leqslant & C \| (w - w' ) P^{t-1}_{\pi}\|_1 \| \pi - \pi' \|_{\infty} + \| (w - w' ) ( P^{t-1}_{\pi} - P^{t-1}_{\pi'} )\|_{1} e^{-\frac{1}{\tau}} \\
            \leqslant & C e^{-\frac{t-1}{\tau}} \| \pi - \pi' \|_{\infty} \| w - w' \|_{1} + C (t-1) e^{- \frac{t-2}{\tau}}\| \pi - \pi' \|_{\infty} \| w - w' \|_{1} e^{-\frac{1}{\tau}} \\
            = & C t e^{- \frac{t-1}{\tau}}\| \pi - \pi' \|_{\infty} \| w - w' \|_{1}.
        \end{aligned}
    \end{equation}
\end{proof}
\par The preceding lemmas will assist us in demonstrating the following main conclusion of the Section \ref{PropertiesofIndividualPayoffFunctions}.
\begin{theorem}
    For any player $i$ and state-action pair $(s,a_i)$, the partial derivative of the value function $\frac{\partial}{\partial \pi_i(a_i|s)} V_i(\pi)$ is Lipschitz continuous with respect to $\pi$. Consequently, the individual payoff gradient $\nabla_{i} V_i(\pi)$ is also Lipschitz continuous.
\end{theorem}

\begin{proof}
    We will prove the theorem by showing that $p_{\pi}(s)$ and $ \overline{\advantage}^{\pi}_i(s,a_i) $ are Lipschitz continuous.
    In the proof of Lemma \ref{Pt-Pt'}, we have
    \begin{equation}
        \begin{aligned}
            \| P_{\pi} - P_{\pi'}  \|_{1} \leqslant (2^N -1) |A| |S| \| \pi - \pi' \|_{\infty}
        \end{aligned}
    \end{equation}
    So we have
    \begin{equation}
        \begin{aligned}
            \| p_\pi - p_{\pi'} \|_1 = &  \| p_{\pi} P_{\pi} - p_{\pi'} P_{\pi} + p_{\pi'} P_{\pi} - p_{\pi'} P_{\pi'}  \|_1 \\
            \leqslant & \| ( p_{\pi} - p_{\pi'} ) P_{\pi}\|_1 + \| p_{\pi'} \|_1 \| P_{\pi} - P_{\pi'}  \|_{1} \\
            \leqslant & e^{-\frac{1}{\tau}} \| ( p_{\pi} - p_{\pi'} )\|_1 + (2^N -1 ) |A| |S| \| \pi - \pi' \|_{\infty},
        \end{aligned}
    \end{equation}
    Thus, we can get that
    \begin{equation}
        \begin{aligned}
            \| p_\pi - p_{\pi'} \|_1 \leqslant \frac{(2^N -1) |A| | S |}{1- e^{-\frac{1}{\tau}}} \| \pi - \pi' \|_{\infty}.
        \end{aligned}
    \end{equation}
    And for $\overline{\advantage}^{\pi}_i(s,a_i)$, we have
    \begin{equation}
        \begin{aligned}
            | \overline{\advantage}^{\pi}_i(s,a_i) - \overline{\advantage}^{\pi'}_i(s,a_i)  | \leqslant & \sum_{a_{-i} \in A_{-i}} | \pi_{-i}(a_{-i}|s) \advantage^\pi_i(s,a_i,a_{-i}) -  \pi'_{-i}(a_{-i}|s) \advantage^{\pi'}_i(s,a_i,a_{-i}) | \\
            \leqslant & \sum_{a_{-i} \in A_{-i}} | \pi_{-i}(a_{-i}|s) [\advantage^\pi_i(s,a_i,a_{-i}) -  \advantage^{\pi'}_i(s,a_i,a_{-i})] | \\
            & +  \sum_{a_{-i} \in A_{-i}} |[\pi_{-i}(a_{-i}|s) - \pi'_{-i}(a_{-i}|s)] \advantage^{\pi'}_i(s,a_i,a_{-i}) | \\
            \triangleq & term1 + term2
        \end{aligned}
    \end{equation}
    For the second term of right side of inequality,
    \begin{equation}
        \begin{aligned}
            term2 =& \sum_{a_{-i} \in A_{-i}} |[\prod_{j \neq i} \pi_j(a_j|s) - \prod_{j \neq i} \pi'_j(a_j|s)]| |\advantage^{\pi'}_i(s,a_i,a_{-i}) | \\
            \leqslant & 2 \max_{s,a} r_i(s,a) \sum_{t=1}^{\infty} e^{-\frac{t}{\tau}} \sum_{a_{-i} \in A_{-i}} (2^{N-1} -1) \| \pi - \pi' \|_{\infty} \\
            \leqslant & 2^N \max_{s,a} r_i(s,a) \sum_{t=1}^{\infty} e^{-\frac{t}{\tau}} \prod_{j \neq i}|A_j|  \| \pi - \pi' \|_{\infty}
        \end{aligned}
    \end{equation}
    \par For the first term, we need to clarify the Lipschitz continuity of $R^{\pi}_i$, by Lemma \ref{prodx}, we have
    \begin{equation}
        \begin{aligned}
           | R^{\pi}_i(s) - R^{\pi'}_i(s) | =& | \sum_{a \in A} (\pi(a|s) - \pi'(a|s)) r_i(s,a) | \\
           \leqslant & \max_{s,a} r_i(s,a) \sum_{a \in A} | (\pi(a|s) - \pi'(a|s)) | \\
           \leqslant & (2^{N}-1) |A| \max_{s,a} r_i(s,a) \| \pi - \pi' \|_{\infty}.
        \end{aligned}
    \end{equation}
    Thus by\eqref{advantage}, we can know that
    \begin{equation}
        \begin{aligned}
            |\advantage^\pi_i(s,a) -  \advantage^{\pi'}_i(s,a)| \leqslant & |p_{\pi} R_i^{\pi} - p_{\pi'} R_i^{\pi'} | + | \sum_{t=1}^{\infty} (e_s - p_{\pi}) P_{\pi}^t R_i^{\pi} -  \sum_{t=1}^{\infty} (e_s - p_{\pi'}) P_{\pi'}^t R_i^{\pi'} | \\
            \triangleq & termA+ termB.
        \end{aligned}
    \end{equation}
    For termA, we have
    \begin{equation}
        \begin{aligned}
            termA \leqslant & | (p_{\pi} - p_{\pi'}) R^\pi_i | + | p_{\pi'}(R^{\pi}_i-R^{\pi'}_i) | \\
            \leqslant & \| (p_{\pi} - p_{\pi'}) \|_{1} \| R^\pi_i \|_{\infty} + \| p_{\pi'} \|_{1} \|(R^{\pi}_i-R^{\pi'}_i) \|_{\infty} \\
            \leqslant & \frac{(2^N-1) |A| |S| \max_{s,a} r_i(s,a)}{1 - e^{-\frac{1}{\tau}}}  \| \pi - \pi' \|_{\infty}
             + (2^N-1) |A| \max_{s,a} r_i(s,a) \| \pi - \pi' \|_{\infty}\\
            = & (2^N-1) |A| \max_{s,a} r_i(s,a)(1 + \frac{|S|}{1 - e^{-\frac{1}{\tau}}}) \| \pi - \pi' \|_{\infty}.
        \end{aligned}
    \end{equation}
 \par For termB, we have
 \begin{equation}\label{decompositionoftermB}
     \begin{aligned}
         termB \leqslant & | \sum_{t=1}^{\infty} (p_{\pi'} - p_{\pi}) P_{\pi}^t R^{\pi}_i | + | 
\sum_{t=1}^{\infty} (e_s - p_{\pi'})  (P_{\pi}^t - P_{\pi'}^t) R^{\pi}_i | \\
        & + | \sum_{t=1}^{\infty} (e_s - p_{\pi'})  P_{\pi'}^t ( R^{\pi}_i - R^{\pi'}_i) |
     \end{aligned}
 \end{equation}
 For the first term of the right side of the inequality \eqref{decompositionoftermB}, we have
 \begin{equation}
     \begin{aligned}
         | \sum_{t=1}^{\infty} (p_{\pi'} - p_{\pi}) P_{\pi}^t R^{\pi}_i | \leqslant & \sum_{t=1}^{\infty} \| (p_{\pi'} - p_{\pi}) P_{\pi}^t \|_{1} \| R^{\pi}_i \|_{\infty} \\
         \leqslant & \max_{s,a} r_i(s,a) \| p_{\pi'} - p_{\pi} \|_1 \sum_{t=1}^{\infty} e^{-\frac{t}{\tau}} \\
         \leqslant & \max_{s,a} r_i(s,a) ( \sum_{t=1}^{\infty} e^{-\frac{t}{\tau}}) \frac{(2^N -1) |A||S|}{1- e^{-\frac{1}{\tau}}} \| \pi - \pi' \|_{\infty} \\
         \triangleq & \alpha_1 \| \pi - \pi' \|_{\infty}
     \end{aligned}
 \end{equation}
For the second term of the right side of the inequality \eqref{decompositionoftermB}, according to Lemma \ref{Pt-Pt'}, we have
 \begin{equation}
     \begin{aligned}
         | \sum_{t=1}^{\infty} (e_s - p_{\pi'})  (P_{\pi}^t - P_{\pi'}^t) R^{\pi}_i | \leqslant & \sum_{t=1}^{\infty} \| (e_s - p_{\pi'})  (P_{\pi}^t - P_{\pi'}^t) \|_{1} \| R^{\pi}_i \|_{\infty} \\
         \leqslant & C \| R^{\pi}_i \|_{\infty} \| (e_s - p_{\pi'}) \|_{1} ( \sum_{t=1}^{\infty} t e^{-\frac{t-1}{\tau}} ) \| \pi - \pi' \|_{\infty} \\
         \leqslant & C |S| \| R^{\pi}_i \|_{\infty}  ( \sum_{t=1}^{\infty} t e^{-\frac{t-1}{\tau}} ) \| \pi - \pi' \|_{\infty} \\
         \triangleq & \alpha_2 \| \pi - \pi' \|_{\infty}.
     \end{aligned}
 \end{equation}
For the third term of the right side of the inequality \eqref{decompositionoftermB}, we have
 \begin{equation}
     \begin{aligned}
         | \sum_{t=1}^{\infty} (e_s - p_{\pi'})  P_{\pi'}^t ( R^{\pi}_i - R^{\pi'}_i) | \leqslant & \sum_{t=1}^{\infty} \|(e_s - p_{\pi'})  P_{\pi'}^t \|_1 \| R^{\pi}_i - R^{\pi'}_i \|_{\infty} \\
         \leqslant & \| (e_s - p_{\pi'}) \|_{1} ( \sum_{t=1}^{\infty} e^{-\frac{t}{\tau}} ) (2^{N}-1) |A| \max_{s,a} r_i(s,a) \| \pi - \pi' \|_{\infty} \\
         \triangleq & \alpha_3 \| \pi - \pi' \|_{\infty}
     \end{aligned}
 \end{equation}
Combining the results above,
\begin{align}
    termB \leqslant (\alpha_1 + \alpha_2 + \alpha_3) \| \pi - \pi' \|_{\infty}.
\end{align}
As a result
\begin{equation}
\begin{aligned}
    |\advantage^\pi_i(s,a) -  \advantage^{\pi'}_i(s,a)| \leqslant & \bigg(\alpha_1 + \alpha_2 + \alpha_3 + (2^N-1) |A| \max_{s,a} r_i(s,a)(1 + \frac{|S|}{1 - e^{-\frac{1}{\tau}}}) \bigg) \| \pi - \pi' \|_{\infty} \\
    \triangleq & \beta \| \pi - \pi' \|_{\infty}.
\end{aligned}
\end{equation}
Apply this result to term1, we have
\begin{equation}
\begin{aligned}
    term1 \leqslant & \sum_{a_{-i} \in A_{-i}} \pi_{-i}(a_{-i}| s) \max_{s,a} |\advantage^\pi_i(s,a) -  \advantage^{\pi'}_i(s,a)| \\
    \leqslant & \beta \prod_{j \neq i} |A_j| \| \pi - \pi' \|_{\infty}.
\end{aligned}
\end{equation}
So there exists $\gamma > 0$, such that
\begin{align}
    | \overline{\advantage}^{\pi}_i(s,a_i) - \overline{\advantage}^{\pi'}_i(s,a_i)  | \leqslant & \gamma \| \pi - \pi' \|_{\infty}.
\end{align}
With all results above, we can obtain that
\begin{equation}
    \begin{aligned}
        |p_{\pi}(s) \overline{\advantage}^\pi_i(s,a_i) - p_{\pi'}(s) \overline{\advantage}^{\pi'}_i(s,a_i)| \leqslant & | (p_{\pi}(s) - p_{\pi'}(s)) \overline{\advantage}^\pi_i(s,a_i) | \\
        & + | p_{\pi'}(s)) ( \overline{\advantage}^\pi_i(s,a_i) - \overline{\advantage}^{\pi'}_i(s,a_i)  ) | \\
        \leqslant & \bigg( 2 (\max_{s,a} r_i(s,a) )\sum_{t=0}^{\infty} e^{-\frac{t}{\tau}} \frac{(2^N -1) |A|}{1- e^{-\frac{1}{\tau}}}+ \gamma \bigg) \| \pi - \pi' \|_{\infty}
    \end{aligned}
\end{equation}
this can show that $\frac{\partial}{\partial \pi_i(a_i|s)} V_i(\pi) = p_{\pi}(s) \overline{\advantage}^\pi_i(s,a_i)$ is Lipschitz continuous. And by Remark \ref{relationaboutnabla}, we can immediately observe the Lipschitz smoothness of $v_i(\pi)$.
\end{proof}

\begin{proof}[Proof of Theorem \ref{Gradientdominanceproperty}]
    We can observe that
    \begin{equation}
        \begin{aligned}
            \advantage^{\pi}_i(s,a) = r_i(s,a) - V_i(\pi) + \sum_{s' \in S} \mathbb{P}[s' | s,a] \sum_{a'_i \in A_i} \sum_{a'_{-i} \in A_{-i}} \pi_i(a'_i| s') \pi_{-i}(a'_{-i}| s') adv^{\pi}_{i}(s',a'_i, a'_{-i}).
        \end{aligned}
    \end{equation}
    Thus, with some direct calculations, denoting $\pi' = (\pi'_i, \pi_{-i})$, we can obtain
    \begin{equation}
        \begin{aligned}
            V_i(\pi'_{i}, \pi_{-i}) - V_i(\pi_{i}, \pi_{-i}) = & \sum_{s \in S} p_{\pi'}(s) \sum_{a_i \in A_i} \sum_{a_{-i} \in A_{-i}}\pi'_i(a_i| s) \pi_{-i}(a_{-i}| s) \big( r_i(s, a_i, a_{-i}) - V_i(\pi_i, \pi_{-i}) \big) \\
            = & \sum_{s \in S} p_{\pi'}(s) \sum_{a_i \in A_i} \sum_{a_{-i} \in A_{-i}}\pi'_i(a_i| s) \pi_{-i}(a_{-i}| s) \advantage^{\pi}_i(s,a_i,a_{-i}) \\
            & - \sum_{s \in S} p_{\pi'}(s) \sum_{(a_i,a_{-i}) \in A} \pi'(a| s)\sum_{s' \in S} \mathbb{P}[s' | s,a] \sum_{(a'_i,a'_{-i}) \in A}  \pi(a'| s') adv^{\pi}_{i}(s',a'_i, a'_{-i}) \\
            = & \sum_{s \in S} p_{\pi'}(s) \sum_{a_i \in A_i}  \pi'_i(a_i| s) \overline{\advantage}^{\pi}_i(s,a_i) - \sum_{s' \in S} p_{\pi'}(s') \sum_{a'_i \in A_i}  \pi_i(a'_i| s') \overline{\advantage}^{\pi}_i(s',a'_i) \\
            = & \sum_{s \in S} p_{\pi'}(s) \sum_{a_i \in A_i}  (\pi'_i(a_i| s)- \pi_i(a_i|s)) \overline{\advantage}^{\pi}_i(s,a_i) \\
            = & \sum_{s \in S} \frac{p_{\pi'}(s)}{p_{\pi}(s)} p_{\pi}(s)\sum_{a_i \in A_i}  (\pi'_i(a_i| s)- \pi_i(a_i|s)) \overline{\advantage}^{\pi}_i(s,a_i) \\
            \leqslant & \| \frac{p_{\pi'}}{p_{\pi}} \|_{\infty} \sum_{s \in S}  \sum_{a_i \in A_i} p_{\pi}(s)  (\pi'_i(a_i| s)- \pi_i(a_i|s)) \overline{\advantage}^{\pi}_i(s,a_i) \\
            \leqslant & C_{\mathcal{G}} \langle \nabla_i V_{i}(\pi), \pi'_i - \pi_i \rangle
        \end{aligned}
    \end{equation}
\end{proof}

\begin{proof}[Proof of Theorem \ref{FirstorderstationarypoliciesareNash}]
    If $\pi^*$ is Nash equilibrium and $\pi = (\pi_i)_{i \in N}$ is any policy profile, then by the definition of Nash equilibrium, for any player $i$, we have
    \begin{equation}
        \begin{aligned}
            \langle v_i(\pi^*), \pi_i -\pi^*_i \rangle \leqslant 0,
        \end{aligned}
    \end{equation}
    So, we can obtain
    \begin{equation}
        \begin{aligned}
            \langle v(\pi^*), \pi -\pi^* \rangle = \sum_{i \in N} \langle v_i(\pi^*), \pi_i -\pi^*_i \rangle \leqslant 0.
        \end{aligned}
    \end{equation}
    And when $\pi^*$ is first-order stationary, if it is not a Nash equilibrium, then there exist a player i and her policy $\pi_i$, such that
    \begin{equation}
        \begin{aligned}
            0 < V_i(\pi_{i}, \pi^*_{-i}) - V_i(\pi^*_{i}, \pi^*_{-i}) \leqslant C_{\mathcal{G}} \langle \nabla_i V_{i}(\pi^*), \pi_i - \pi^*_i \rangle,
        \end{aligned}
    \end{equation}
    So we have
    \begin{equation}
        \begin{aligned}
            \langle \nabla_i V_{i}(\pi^*), \pi_i - \pi^*_i \rangle > 0,
        \end{aligned}
    \end{equation}
    And let $\pi = (\pi_i, \pi^*_{-i})$, we have
    \begin{equation}
        \begin{aligned}
            \langle v(\pi^*), \pi -\pi^* \rangle > 0,
        \end{aligned}
    \end{equation}
    which leads to the contradiction!
\end{proof}

\section{Proof of Section \ref{TheLearningFramework}}

\par To prove Lemma \ref{propertyofFenchel}, we present the following lemma directly, which is essentially folklore in optimization. For a detailed proof, readers are referred to \cite{rockafellar2009variational}.

\begin{lemma}\label{convexproperty}
    Let $h$ be a $K$-strongly convex regularizer with induced mirror map $Q$ and convex conjugate $h^*$, then we have
    \begin{enumerate}
        \item $x = Q(y)$ if and only if $y \in \partial h(x)$.
        \item $h^*$ is differentiable on $\mathcal{Y}$ and $\nabla h^*(y) = Q(y)$ for all $y \in \mathcal{Y}$.
        \item $Q$ is $1/K$-Lipschitz continuous.
    \end{enumerate}
\end{lemma}
\par Then, let 
\begin{equation}
    D(p,x) = h(p) - h(x) - \langle \nabla h(x) , p-x \rangle,
\end{equation}
we have
\begin{lemma}\label{Q=D}
    For all $p \in \mathcal{X}$ and all $y \in \mathcal{Y}$, we have
    \begin{equation}
        F(p,y) = D(p, Q(y)), \quad \text{ if } Q(y) \in \mathcal{X}^{\circ}.
    \end{equation}
\end{lemma}
\begin{proof}
    \par By definition \eqref{definitionofQ}
    \begin{equation}
        \begin{aligned}
            F(p,y) =& h(p) + \langle y, Q(y) \rangle - h(Q(y)) - \langle y,p \rangle, \\
            =& h(p) - h(Q(y)) - \langle y, p - Q(y) \rangle.
        \end{aligned}
    \end{equation}
    Since $y \in \partial h(Q(y))$ by Lemma \ref{convexproperty}, we have $\langle y, p - Q(y) \rangle = \langle \nabla h(Q(y)), p - Q(y) \rangle$, which leads to the conclusion directly.
\end{proof}

\begin{proof}[Proof of Lemma \ref{propertyofFenchel}]
    Let $x = Q(y)$. By the definition of $K$-strongly convex, we have
    \begin{equation}
    \begin{aligned}
        h(x) + \lambda \langle y, p-x \rangle \leqslant & h(x+ \lambda(p-x)) \\
        \leqslant & \lambda h(p) + (1-\lambda)h(x) - \frac{K}{2} \lambda (1-\lambda) \| x-p \|^2.
    \end{aligned}
    \end{equation}
    when $t \in ( 0, 1 ]$, we have
    \begin{equation}
        \frac{K}{2}(1-\lambda) \| x-p \|^2 \leqslant h(p) - h(x) - \langle y, p-x \rangle = F(p,y).
    \end{equation}
    The first conclusion of Lemma \ref{propertyofFenchel} then follows by letting $ \lambda \to 0^+ $.
    As a result, we have
    \begin{equation}
        \frac{K}{2} \| x-p \|^2 \leqslant D(p,x).
    \end{equation}
    Let $x' = Q(y')$. By definition \eqref{definitionofF} and Lemma \ref{convexproperty},
    \begin{equation}
        \begin{aligned}
            F(p,y') =& h(p) + h^*(y') - \langle y', p \rangle \\
            =& h(p) + \langle y', Q(y') - p \rangle - h(Q(y')) \\
            =& F(p,y) + \langle y' -y , Q(y)-p \rangle + \langle y', Q(y') - Q(y)\rangle + h(Q(y)) - h(Q(y')) \\
            = & F(p,y) + \langle y' -y , Q(y)-p \rangle + \langle y' - y, x' - x\rangle - D(x', x).
        \end{aligned}
    \end{equation}
    By Young's inequality, we have
    \begin{equation}
        \langle y' - y, x' - x\rangle \leqslant \frac{K}{2} \| x' - x \|^2 + \frac{1}{2K} \|y' - y\|^2.
    \end{equation}
    By Lemma \ref{Q=D} we obtain that
    \begin{equation}
        \begin{aligned}
            \langle y' - y, x' - x\rangle - D(x', x) \leqslant & \frac{K}{2} \| x' - x \|^2 + \frac{1}{2K} \|y' - y\|^2 - D(x', x) \\
            \leqslant & \frac{1}{2K} \|y' - y\|^2.
        \end{aligned}
    \end{equation}
    Thus, we have
    \begin{equation}
        F(p,y') \leqslant F(p,y) + \langle y' -y , Q(y)-p \rangle + \frac{1}{2K} \|y' - y\|^2.
    \end{equation}
\end{proof}

\begin{proof}[Proof of Lemma \ref{estamitorproperty}]
    By the independence of the sampling directions $(z_i)_{i \in N}$, we have
    \begin{equation*}
        \begin{aligned}
            \mathbb{E}[\hat{v}_i] = & \frac{d_i / \delta}{\prod_j vol(\mathbb{S}_j)} \int_{\mathbb{S}_1} \cdots \int_{\mathbb{S}_N} V_i(\pi_1 + \delta z_1,\dots, \pi_N+ \delta z_N) z_i \mathrm{d}z_1 \cdots \mathrm{d} z_N \\
            = & \frac{d_i / \delta}{\prod_j vol(\delta \mathbb{S}_j)} \int_{\delta \mathbb{S}_i} \int_{\prod_{j \neq i} \delta\mathbb{S}_j} V_i(\pi_i +  z_i, \pi_{-i}+  z_{-i}) \frac{z_i}{\| z_i \|} \mathrm{d}z_i \mathrm{d} z_{-i} \\
            = & \frac{d_i / \delta}{\prod_j vol(\delta \mathbb{S}_j)} \int_{\delta \mathbb{B}_i} \int_{\prod_{j \neq i} \delta\mathbb{S}_j} \nabla_i V_i(\pi_i + w_i,\dots, \pi_{-i}+ z_{-i}) \mathrm{d}w_i  \mathrm{d} z_{-i}.
        \end{aligned}
    \end{equation*}
    where, in the last equality, we use the Stoke's theorem
    \begin{equation*}
        \begin{aligned}
            \nabla \int_{\delta \mathbb{B}} f(x+w) \mathrm{d}w = \int_{\delta \mathbb{S}} f(x+z) \frac{z}{\|z\|} \mathrm{d}z.
        \end{aligned}
    \end{equation*}
    Since $vol(\delta \mathbb{B}_i) = \frac{\delta}{d_i} vol(\delta \mathbb{S}_i)$, we obtain that $\mathbb{E}[\hat{v}_i] = \nabla_i V^{\delta}_i (\pi_i, \pi_{-i})$.
    \par When $ v_i $ is Lipschitz continuous, let $L_i$ be the Lipschitz constnt of $v_i$, i.e.,
    \begin{align*}
        \| v_i(\pi') - v_i(\pi) \| \leqslant L_i \| \pi' - \pi \|, \forall \pi, \pi' \in \Pi,
    \end{align*}
    then for any $w_i \in \delta \mathbb{B}$ and any $z_j \in \mathbb{S}_j, j \neq i$, we have
    \begin{align*}
        \| \nabla_i V_i(\pi_i+ w_i, \pi_{-i}+ z_{-i}) - \nabla_i V_i(\pi) \| \leqslant L_i \sqrt{ \|w_i\|^2 + \sum_{j \neq i}\|z_j\|^2 } \leqslant L_i \sqrt{N} \delta.
    \end{align*}
    The second part of lemma can be obtained by integrating and differentiating under the integral sign.
\end{proof}

\begin{proof}[Proof of Lemma \ref{differenceinvalue}]
    \par It's easy to check that
    \begin{align*}
        \mathbb{E}[\hat{V}^t_i] = e_{s^t} P_{\hat{x}^t}^{T^t} R^{\hat{x}^t}_i,
    \end{align*}
where $R^{\hat{x}^t}_i$ is the expected one-step reward vector defined in the proof of Lemma \ref{advbound}. So by Assumption \ref{assumptionintransition}, we obtain that
    \begin{equation}
        \begin{aligned}
            | \mathbb{E}[\hat{V}^t_i] - V_i(\hat{x}^t) | =& | e_{s^t} P_{\hat{x}_t}^{T^t} R^{\hat{x}^t}_i - p_{\hat{x}^t} R^{\hat{x}^t}_i| \\
            \leqslant & (\max_{s,a} r_i(s,a)) \| e_{s^t} P_{\hat{x}_t}^{T^t} - p_{\hat{x}^t} P_{\hat{x}_t}^{T^t} \|_1 \\
            \leqslant & (\max_{s,a} r_i(s,a)) \| e_{s^t}  - p_{\hat{x}^t}  \|_1 e^{-\frac{T^t}{\tau}} \\
            \leqslant & |S| (\max_{s,a} r_i(s,a)) e^{-\frac{T^t}{\tau}}.
        \end{aligned}
    \end{equation}
\end{proof}

\section{Proof of Section \ref{Convergenceanalysisandresults}}

\par The proof of Theorem \ref{convergencetheorem} is divided into three steps as follows. First, for any globally neutrally stable Nash equilibrium $\pi^*$, we will demonstrate that $F(\pi^*, Y^t)$ converges to a finite random variable $F^{\infty}$. Then we will show that there exists a subsequence $\{ \pi^{t_k} \}$ which will converges to some Nash equilibrium if there exists a globally stable Nash equilibrium. Then the proof of Theorem \ref{convergencetheorem} can be completed with these two facts as well as the properties of the Fenchel coupling.

\begin{lemma}\label{step1}
    Let $\pi^*$ be a globally neutrally stable Nash equilibrium of $\mathcal{G}$, then with assumptions as in Theorem \ref{convergencetheorem}, the sequence $\{ F(\pi^*, Y^t) \}$ converges to a finite random variable $F^{\infty}$ almost surely.
\end{lemma}

\begin{proof}
    Let $K = \min_{i \in N} K_i$, by Lemma \ref{propertyofFenchel}, we have
    \begin{equation}
        \begin{aligned}
            F(\pi^*, Y^{t+1}) \leqslant & F(\pi^*, Y^t) + \gamma^t \langle \hat{v}^t , \pi^t - \pi^* \rangle + (\frac{\gamma^t}{2K})^2 \| \hat{v}^t \|^2 \\
            = & F(\pi^*, Y^t) + \gamma^t \langle v^t + b^t + U^t + \epsilon^t , \pi^t - \pi^* \rangle + (\frac{\gamma^t}{2K})^2 \| \hat{v}^t \|^2 \\
            \leqslant & F(\pi^*, Y^t) + \gamma^t \langle b^t + U^t + \epsilon^t , \pi^t - \pi^* \rangle + \frac{(\gamma^t)^2}{2K} \| \hat{v}^t \|^2 .
        \end{aligned}
    \end{equation}
    to derive the final inequality, we use the Definition \ref{monotone} of globally neutrally stable policy profiles. By \eqref{definitionofhatv}, we can get the upper bound of $\| \hat{v}^t \|^2$
    \begin{equation}
        \| \hat{v}^t \|^2 = \sum_{i \in N} \| \frac{|S|(|A_i|-1)}{\delta^t} \hat{V}^t_i \cdot F_i z^t_i \|^2 \leqslant \frac{1}{(\delta^t)^2} \sum_{i \in N} |S|(|A_i|-1) \max_{z_i^t \in \mathbb{S}_i} \| F_i z_i^t \|^2.
    \end{equation}
    As $\mathbb{S}_i$ is a bounded closed set, $\max_{z_i^t \in \mathbb{S}_i} \| F_i z_i^t \|^2 < \infty$, thus there exists a constant $U > 0$ such that 
    \begin{equation}
        \| \hat{v}^t \|^2 \leqslant (\frac{U}{\delta^t})^2.
    \end{equation}
    \par Now, conditioning on the history $\mathcal{F}^t$ of $\pi^t$ and taking expectations, we obtain
    \begin{equation}
        \begin{aligned}
            \mathbb{E}[F(\pi^*, Y^{t+1})| \mathcal{F}^t] \leqslant F(\pi^*, Y^t) + \gamma^t\mathbb{E}[\langle b^t + U^t + \epsilon^t , \pi^t - \pi^* \rangle | \mathcal{F}^t] + \frac{U^2}{2K}(\frac{\gamma^t}{\delta^t})^2
        \end{aligned}
    \end{equation}
    Because $\pi^t$ is $\mathcal{F}^t$-maesurable, so
    \begin{equation}
        \mathbb{E}[\langle U^t , \pi^t - \pi^* \rangle | \mathcal{F}^t] =  \langle \mathbb{E}[U^t| \mathcal{F}^t], \pi^t - \pi^* \rangle = 0.
    \end{equation}
    By Theorem \ref{Lipschitzofgradient} and Lemma \ref{estamitorproperty}, we have
    \begin{equation}
        \begin{aligned}
            \| b^t \| =& \sqrt{ \sum_{i \in N} \| F_i \nabla_{i} V_i^{\delta^t}(x^t_{\delta^t})- \nabla_{i} V_i(\pi^t_i, \pi^t_{-i})) \|^2 } \\
            \leqslant& \sqrt{ \sum_{i \in N} \| F_i \|^2 \bigg( \|\nabla_{i} V_i^{\delta^t}(x^t_{\delta^t})- \nabla_{i} V_i(x^t_{\delta^t})\|^2 + \| \nabla_{i} V_i(x^t_{\delta^t})- \nabla_{i} V_i(x^t) \|^2  \bigg) } \\
            =& \mathcal{O}(\delta^t).
        \end{aligned}
    \end{equation}
    By Lemma \ref{differenceinvalue}, we have
    \begin{equation}
        \begin{aligned}
            \| \epsilon^t \| =& \sqrt{ \sum_{i \in N} \| F_i ( \frac{|S|(|A_i|-1)}{\delta^t} \hat{V}^t_i \cdot z^t_i - \frac{|S|(|A_i|-1)}{\delta^t} V_i(\hat{x}^t) \cdot z^t_i) \|^2  } \\
            \leqslant& \frac{1}{\delta^t} \sqrt{ \sum_{i \in N} \| F_i \|^2 |S|^2(|A_i|-1)^2 (\hat{V}^t_i - V_i(\hat{x}^t))^2 } \\
            = & \mathcal{O}(\frac{e^{- \frac{T^t}{\tau}}}{\delta^t}).
        \end{aligned}
    \end{equation}
    So, there exists some $B_1, B_2 > 0$, such that
    \begin{equation}
        \langle b^t + \epsilon^t , \pi^t - \pi^* \rangle \leqslant B_1 \delta^t + B_2 \frac{e^{- \frac{T^t}{\tau}}}{\delta^t}.
    \end{equation}
    As a result, we have
    \begin{equation}
        \begin{aligned}
            \mathbb{E}[F(\pi^*, Y^{t+1})| \mathcal{F}^t] \leqslant F(\pi^*, Y^t) + B_1 \delta^t \gamma^t + B_2 \frac{e^{- \frac{T^t}{\tau}}}{\delta^t} \gamma^t + \frac{U^2}{2K}(\frac{\gamma^t}{\delta^t})^2.
        \end{aligned}
    \end{equation}
    Letting $E^t = F(\pi^*, Y^{t+1}) + \sum_{k = t}^{\infty} [B_1 \delta^t \gamma^t + B_2 \frac{e^{- \frac{T^t}{\tau}}}{\delta^t} \gamma^t + \frac{U^2}{2K}(\frac{\gamma^t}{\delta^t})^2]$, we have
    \begin{equation}
        \begin{aligned}
            \mathbb{E} [E^{t+1} | \mathcal{F}^t] = & \mathbb{E}[F(\pi^*, Y^{t+1})| \mathcal{F}^t] + \sum_{k = t+1}^{\infty} [B_1 \delta^t \gamma^t + B_2 \frac{e^{- \frac{T^t}{\tau}}}{\delta^t} \gamma^t + \frac{U^2}{2K}(\frac{\gamma^t}{\delta^t})^2] \\
            \leqslant & F(\pi^*, Y^{t}) + \sum_{k = t}^{\infty} [B_1 \delta^t \gamma^t + B_2 \frac{e^{- \frac{T^t}{\tau}}}{\delta^t} \gamma^t + \frac{U^2}{2K}(\frac{\gamma^t}{\delta^t})^2] = E^t,
        \end{aligned}
    \end{equation}
    which means that $E^t$ is a $\mathcal{F}^t-$adapted supermartingale. Since as the assumptions of Theorem \ref{convergencetheorem} holds, it follows that
    \begin{equation}
        \begin{aligned}
            \mathbb{E}[E^t] = \mathbb{E} [\mathbb{E}[E^t|\mathcal{F}^{t-1}] ] \leqslant \mathbb{E}[E^{t-1}] \leqslant \cdots \leqslant \mathbb{E}[E^{1}] < \infty,
        \end{aligned}
    \end{equation}
    i.e., $E^t$ is uniformly bounded in $L^1$. Thus, by Doob's convergence theorem for supermartingales, it follows that $E^t$ converges (a.s.) to some finite random variable $E^{\infty}$. As $\sum_{k = t}^{\infty} [B_1 \delta^t \gamma^t + B_2 \frac{e^{- \frac{T^t}{\tau}}}{\delta^t} \gamma^t + \frac{U^2}{2K}(\frac{\gamma^t}{\delta^t})^2]$ converges (a.s.) to $0$, so $F(\pi^*, Y^{t+1})$ converges (a.s.) to some finite random variable $F^{\infty}$.
\end{proof}

\begin{lemma}\label{step2}
    Suppose that the assumptions of Theorem \ref{convergencetheorem} hold. With probability $1$, there exists a subsequence $\pi^{t_k}$ which converges to some Nash equilibrium.
\end{lemma}

\begin{proof}
    \par Suppose the lemma is incorrect, i.e., with a positive probability, the sequence $\pi^t$ has no limit points in the set of Nash equilibria. In this event, and given that the set of Nash equilibria $\Pi^*$ is compact \cite{Rosen}, there exists some compact set $\mathcal{C}$ such that $\mathcal{C} \cap \Pi^* = \emptyset$ and $\pi^t \in \mathcal{C}$ for sufficient large $t$. Moreover, by Definition \ref{monotone} and let $\pi^*$ be the globally variationally stable Nash equilibrium mentions in Theorem \ref{convergencetheorem}, we have $\langle v(\pi), \pi - \pi^* \rangle <0 $, whenever $\pi \in \mathcal{C}$. Thus, by the continuity of $v$ and the compactness of $\Pi^*$ and $\mathcal{C}$, there exists some $c> 0$ such that
    \begin{equation}
        \langle v(\pi), \pi - \pi^* \rangle \leqslant -c, \quad \forall \pi \in \mathcal{C}. 
    \end{equation}
    For the globally variationally stable Nash equilibrium $\pi^* \in \Pi^*$, as is proved in Lemma \ref{step1}, and letting $\eta^t = \sum_{k=0}^{t} \gamma^k$, we have
    \begin{equation}
        \begin{aligned}
            F(\pi^*, Y^{t+1}) \leqslant &  F(\pi^*, Y^t) + \gamma^t \langle v^t + b^t + U^t + \epsilon^t , \pi^t - \pi^* \rangle + (\frac{\gamma^t}{2K})^2 \| \hat{v}^t \|^2 \\
            \leqslant & F(\pi^*, Y^0) - \eta^t \bigg[ 
            c - \frac{\sum_{k=0}^{t} \bigg( \gamma^k \langle b^k + U^k + \epsilon^k , \pi^k - \pi^* \rangle + (\frac{\gamma^k}{2K})^2 \| \hat{v}^k \|^2 \bigg) }{\eta^t}
            \bigg]
        \end{aligned}
    \end{equation}
    By the definition of $U^t$ in \eqref{decomposition}, it is a martingale difference sequence and there exists $\sigma > 0$ such that 
    \begin{equation}
        \| U^{t} \|^2 \leqslant \frac{\sigma^2}{(\delta^t)^2},
    \end{equation}
    hence
    \begin{equation}
        \sum_{t=0}^{\infty}( \gamma^{t+1} \mathbb{E}[ \langle U^{t+1}, \pi^{t+1} - \pi^* \rangle | \mathcal{F}^t] )^2 \leqslant 2 \sigma^2 \sum_{t=0}^{\infty} (\frac{\gamma^t}{\delta^t})^2 < \infty.
    \end{equation}
    Therefore, by the law of large numbers for martingale difference sequences, we have
    \[(\eta^t)^{-1}\sum_{k=1}^{t} \gamma^k \langle U^{k}, \pi^{k} - \pi^* \rangle \]
    converges to $0$ with probability $1$.
    \par Following the similar way, we get
    \[
    (\eta^t)^{-1}\sum_{k=1}^{t} \gamma^k \langle b^k, \pi^{k} - \pi^* \rangle, \quad
    (\eta^t)^{-1}\sum_{k=1}^{t} \gamma^k \langle \epsilon^k, \pi^{k} - \pi^* \rangle
    \]
    converges to $0$ with probability $1$.
    \par Finally, for the term $ \sum_{k=0}^{t} (\frac{\gamma^k}{2K})^2 \| \hat{v}^k \|^2 $, we have
    \begin{equation}
    \begin{aligned}
        \mathbb{E}[ \sum_{k=0}^{t+1} (\frac{\gamma^k}{2K})^2 \| \hat{v}^k \|^2 | \mathcal{F}^t] =& \sum_{k=0}^{t} (\frac{\gamma^k}{2K})^2 \| \hat{v}^k \|^2 + \mathbb{E}[ (\frac{\gamma^{t+1}}{2K})^2 \| \hat{v}^{t+1} \|^2 \mathcal{F}^t ] \\
        \geqslant& \sum_{k=0}^{t} (\frac{\gamma^k}{2K})^2 \| \hat{v}^k \|^2,
    \end{aligned}
    \end{equation}
    i.e., it is a submartingale with respect to $\mathcal{F}^t$. Moreover, as is proved in Lemma \ref{step1}, $ \sum_{k=0}^{t} (\frac{\gamma^k}{2K})^2 \| \hat{v}^k \|^2 $ is uniformly bounded in $L^1$. Thus by Doob's submartingale convergence theorem, $ \sum_{k=0}^{t} (\frac{\gamma^k}{2K})^2 \| \hat{v}^k \|^2 $ converges (a.s.) to some finite random variable. Consequently, we have
    \[
    (\eta^t)^{-1} \sum_{k=0}^{t} (\frac{\gamma^k}{2K})^2 \| \hat{v}^k \|^2 
    \]
    converges to $0$ with probability $0$.
    With all results above, we have $F(\pi^*, Y^{t+1}) \leqslant F(\pi^*, Y^{0}) - c \frac{\eta^t}{2}$ for sufficient large $n$, and hence $F(\pi^*, Y^{t+1}) \to -\infty$, a contradiction to Lemma \ref{step1}!
\end{proof}
\par With the preparations above, now we can prove the Theorem \ref{convergencetheorem}.
\begin{proof}[Proof of Theorem \ref{convergencetheorem}]
\par By Lemma \ref{step2}, there exists (a.s.) a globally neutrally stable Nash equilibrium $\pi^* \in \Pi^*$ such that $Q(Y^{t_k}) \to \pi^*$ for some subsequence $Y^{t_k}$. By the assumption \eqref{reciprocitycondition} of the Fenchel coupling, we have
\begin{equation}
    \lim_{k \to \infty} F(\pi^*, Y^{t_k}) = 0, \quad \text{a.s.}.
\end{equation}
By Lemma \ref{step1}, $\lim_{t \to \infty} F(\pi^*, Y^{t})$ exists almost surely, it follows that
\begin{equation}
    \lim_{t \to \infty} F(\pi^*, Y^{t}) = \lim_{k \to \infty} F(\pi^*, Y^{t_k}) = 0, \quad \text{a.s.}.
\end{equation}
This shows that $\pi^t$ converges to $\pi^*$ with probability $1$ by Lemma \ref{propertyofFenchel}.
    
\end{proof}

\end{document}